\documentclass[a4paper,11pt]{article}

\usepackage{latexsym}
\usepackage{amssymb}
\usepackage{amsfonts}
\usepackage{amsmath}
\usepackage{amsthm}
\usepackage{indentfirst}
\usepackage{xypic}
\usepackage[english]{babel}
\usepackage{cite}

\usepackage{fullpage}
\usepackage{graphicx}
\usepackage{caption}
\usepackage{rotating}
\usepackage{color}
\usepackage{indentfirst}
\usepackage{manfnt}

\usepackage{tikz}
\usepackage[boxruled]{algorithm2e}

\newcommand{\Av}{\mathrm{Av}}
\newcommand{\pat}[1]{\left\lfloor\begin{smallmatrix}#1\end{smallmatrix}\right\rfloor}
\newcommand{\Sort}{\mathrm{Sort}}
\newcommand{\DSort}{\mathrm{Sort}^{\downarrow}}
\newcommand{\U}{\mathtt{U}} 
\newcommand{\D}{\mathtt{D}} 
\renewcommand{\H}{\mathtt{H}} 

\newtheorem{theorem}{Theorem}[section]
\newtheorem{prop}[theorem]{Proposition}
\newtheorem{lemma}[theorem]{Lemma}

\newtheorem{corollary}[theorem]{Corollary}

\date{}

\author{Giulio Cerbai\thanks{Dipartimento di Matematica e Informatica ``U.
Dini", University of Firenze, Firenze, Italy,
\tt{giulio.cerbai@unifi.it, luca.ferrari@unifi.it}}
\and Anders Claesson\thanks{Science Institute, University of Iceland, Iceland, \tt{akc@hi.is}}
\and Luca Ferrari$^{\dag}$}

\title{Stack sorting with restricted stacks\footnote{G.C. and L.F. are members of the INdAM Research group GNCS; they are partially supported by INdAM - GNCS 2019 project ``Studio di propriet\'a combinatoriche di linguaggi formali ispirate dalla biologia e da strutture bidimensionali" and by a grant of the "Fondazione della Cassa di Risparmio di Firenze" for the project "Rilevamento di pattern: applicazioni a memorizzazione basata sul DNA, evoluzione del genoma, scelta sociale".}}

\begin{document}

\maketitle

\begin{abstract}
  The (classical) problem of characterizing and enumerating permutations that can be sorted using two stacks connected in series is still largely open.  In the present paper we address a related problem, in which we impose restrictions both on the procedure and on the stacks.
  More precisely, we consider a \emph{greedy} algorithm where we perform the rightmost legal operation
  (here "rightmost" refers to the usual representation of stack sorting problems).
  Moreover, the first stack is required to be $\sigma$-avoiding, for some permutation $\sigma$,
  meaning that, at each step, the elements maintained in the stack avoid the pattern $\sigma$ when read from top to bottom.
  Since the set of permutations which can be sorted by such a device (which we call \emph{$\sigma$-machine}) is not always a class,
  it would be interesting to understand when it happens.
  We will prove that the set of $\sigma$-machines whose associated sortable permutations are not a class is counted by Catalan numbers.
  Moreover, we will analyze two specific $\sigma$-machines in full details (namely when $\sigma =321$ and $\sigma =123$),
  providing for each of them a complete characterization and enumeration of sortable permutations.
\end{abstract}

\section{Introduction}

The birth of \emph{stack sorting disciplines} can be traced back to a series of exercises proposed by Knuth in \cite{Kn}.
Consider the problem of sorting a permutation $\pi =\pi_1 \pi_2 \cdots \pi_n$ using a stack.
More specifically, scan the permutation from left to right and, when $\pi_i$ is read,
either push $\pi_i$ onto the stack or pop the top of the stack into the output (at the first available position).
Call the two above operations $S$ and $O$, respectively.
Knuth has showed that there is an optimal algorithm, called \emph{Stacksort},
which is able to sort every sortable permutation.
It consists of performing $S$ whenever $\pi_i$ is smaller than the current top of the stack,
otherwise performing $O$ (see Listing~\ref{stacksort}).

\begin{algorithm}
$Stack:=\emptyset$\;
\While{$i\leq n$}
{
    \If{$Stack=\emptyset$ or $\pi_i < TOP(Stack)$}
    {
        $\textnormal{execute}$ S\;
        $i:=i+1$\;
    }
    \Else{$\textnormal{execute}$ O\;}
}
\While{$Stack \neq \emptyset$}
{$\textnormal{execute}$ O\;}
\caption{Stacksort ($Stack$ is the stack, $TOP(Stack)$ is the current top of the stack,
$\pi=\pi_1 \cdots \pi_n$ is the input
permutation).}\label{stacksort}
\end{algorithm}

It is easy to realize that \emph{Stacksort} has two key properties:
\begin{enumerate}
\item the stack is \emph{increasing},
meaning that the elements inside the stack are maintained in increasing order (from top to bottom);
\item the algorithm is \emph{right greedy},
meaning that it always chooses to perform $S$ as long as the stack remains increasing in the above sense;
here the expression ``right greedy" refers to the usual pictorial representation of this problem,
in which the input permutation is on the right,
the stack is in the middle and the output permutation is on the left
(see Figure~\ref{stacksort_machine}, on the left).
\end{enumerate}

Using \emph{Stacksort}, it can be shown that sortable permutations are precisely those avoiding the pattern 231;
and it is well known that 231-avoiding permutations of length $n$ are counted by Catalan numbers.

\bigskip

Though the above problem is rather simple,
things become considerably more complicated if one allows more stacks connected in series.
As a matter of fact, for the machine consisting of just two stacks in series we know at present very few results.
We know, for instance,
that sortable permutations can be characterized in terms of an \emph{infinite} set of avoided patterns,
but we do not have any explicit description of such a set \cite{M}.
Needless to say, the enumeration of sortable permutations is completely unknown.

Since the general problem of sorting with two stacks is too difficult, several special cases have been considered.  Among them, the so-called \emph{West-2-stack-sortable permutations} \cite{W} are those permutations
which can be sorted by making two passes through a stack.
Equivalently, they are the permutations that can be sorted by 2 stacks connected in series
using a \emph{right greedy algorithm} (see \cite{W} for more details).
West-2-stack-sortable permutations do not form a class,
nevertheless it is possible to characterize them using some kind of generalized patterns
(called \emph{barred patterns}).

Another possible variation on the two-stacks problem is to impose some restrictions on the content of the stack.
Rebecca Smith \cite{Sm} has studied the case in which the first stack is required to be \emph{decreasing}.
Notice that, if we do not choose a specific algorithm in advance,
the second stack turns out to be necessarily increasing.
In the above case, Smith is able to describe an optimal sorting algorithm,
thanks to which she can completely characterize (in terms of avoided patterns) and enumerate sortable permutations.

\bigskip

In the present paper we will deal with similar sorting machines
consisting of two stacks connected in series (see Figure~\ref{stacksort_machine}, on the right).

\begin{figure}[h!]
\begin{center}
\includegraphics[viewport=0 700 500 800]{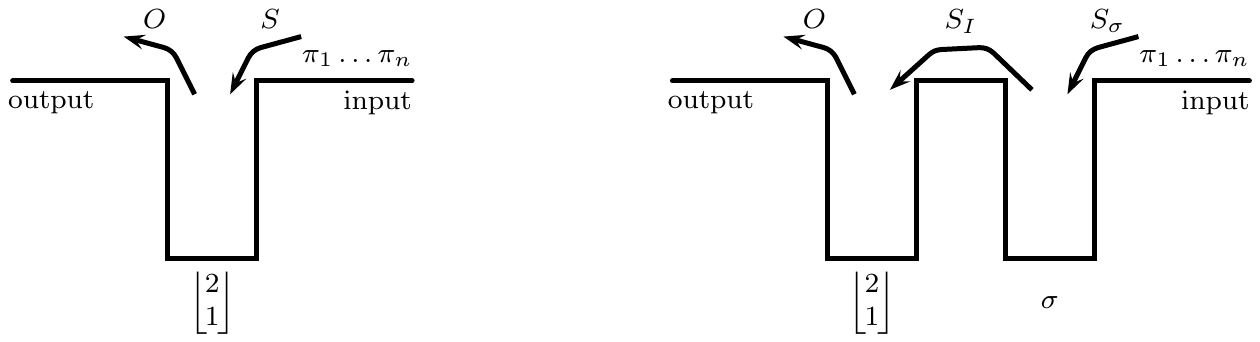}
\end{center}
\caption{Sorting with one stack (on the left) and sorting with two stacks, where the first one is $\sigma$-restricted (on the right).}
\label{stacksort_machine}
\end{figure}

Recalling the key properties of the \emph{Stacksort} algorithm,
we will consider machines obeying certain constraints, which are described below.

\begin{enumerate}

\item The stacks must obey some \emph{restrictions}, which are expressed by saying that,
at each step of the execution,
the elements into each stack (read from top to bottom) must avoid certain forbidden configurations.
In particular, in analogy with \emph{Stacksort}, we require the second stack to be increasing.
Notice that this can be equivalently expressed as follows:
at every step,
the sequence of numbers contained in the stack (read from top to bottom) has to avoid the pattern 21.
We will express this by saying that the stack is $\pat{2\\1}$-avoiding.
Moreover, we will be interested in machines in which the first stack is \emph{$\sigma$-avoiding},
for some pattern $\sigma$.

\item The algorithm we perform on the two stacks connected in series is \emph{right greedy}.
As already observed, this is equivalent to making two passes through a stack,
performing the right greedy algorithm at each pass.
However, due to the restriction described above, during the first pass the stack is $\sigma$-avoiding,
whereas during the second pass it is $\pat{2\\1}$-avoiding.

\end{enumerate}

We will use the term \emph{$\sigma$-machine} to refer to
the right greedy algorithm performed on two stacks in series,
such that the first stack is $\sigma$-avoiding and the second stack is $\pat{2\\1}$-avoiding.
Formally, the algorithm we are going to analyze is described in Listing~\ref{sigma-avoiding}.

\begin{algorithm}
$Stack_I:=\emptyset$\;
$Stack_{\sigma}:=\emptyset$\;
$i:=1$\;
\While{$i\leq n$}
{
    \If{$\sigma \nleq Stack_{\sigma}\circ \pi_i$}
    {
        $\textnormal{execute}$ $S_{\sigma}$\;
        $i:=i+1$\;
    }
    \ElseIf{$Stack_I =\emptyset$ or $TOP(Stack_{\sigma})<TOP(Stack_I )$}
    {$\textnormal{execute}$ $S_{I}$\;}
    {$\textnormal{execute}$ O\;}
}
\While{$Stack_{\sigma}\neq \emptyset$}
{
    \eIf{$Stack_I =\emptyset$ or $TOP(Stack_{\sigma})<TOP(Stack_I )$}
    {$\textnormal{execute}$ $S_{I}$\;}
    {$\textnormal{execute}$ O\;}
}
\While{$Stack_I\neq \emptyset$}
{$\textnormal{execute}$ O\;}
\caption{The $\sigma$-machine ($Stack_{\sigma}$ is the $\sigma$-avoiding stack, $Stack_I$ is the increasing stack,
$S_{\sigma}$ means pushing into $Stack_{\sigma}$, $S_I$ means pushing into $Stack_I$,
O means moving $TOP(Stack_I )$ into the output, $\circ$ is the concatenation operation).}\label{sigma-avoiding}
\end{algorithm}

The set of permutations which are sortable by the $\sigma$-machine is denoted $\Sort(\sigma)$
and its elements are the \emph{$\sigma$-sortable permutations}.
The set of $\sigma$-sortable permutations of length $n$ is denoted $\Sort_n (\sigma)$.
In the present paper we initiate the study of the combinatorics of $\sigma$-machines.
In particular, we aim at characterizing and enumerating $\sigma$-sortable permutations.
After necessary preliminaries (contained in Section~\ref{prelim}), we
easily realize that the set of $\sigma$-sortable permutations is a
permutation class for some choices of $\sigma$, while it is not a
permutation class for other choices of $\sigma$.
In Section~\ref{classes} we find an explicit characterization of those $\sigma$ such that
$\sigma$-sortable permutations constitute a class,
and prove the striking fact that $\sigma$-machines whose $\sigma$-sortable permutations are not a class
are counted by Catalan numbers (with respect to the length of $\sigma$).
Then we will focus on a couple of specific $\sigma$-machines:
Section~\ref{321machine} studies the $321$-machine,
giving a complete characterization and enumeration of sortable permutations, which happen to constitute a class
(our result is actually stronger, being stated for a decreasing permutation $\sigma$ of any length);
Section~\ref{123machine} is devoted to the analysis of the $123$-machine,
and also in this (considerably more challenging) case
we are able to provide complete structural and enumerative results for sortable permutations
(which do not form a class, by the way),
describing in particular a bijection with a specific set of pattern-avoiding Schr\"oder paths.
The last section suggests some directions for further research.

\section{Preliminaries and notations}\label{prelim}

%

Given a permutation $\pi =\pi_1 \pi_2 \cdots \pi_n$ of length $n$,
the \emph{$k$-inflation of $\pi$ at $\pi_i$} is the permutation of length $n+(k-1)$
obtained from $\pi$ by replacing $\pi_i$ with the consecutive increasing sequence
$\pi_i (\pi_{i}+1)\cdots (\pi_{i}+(k-1))$ and suitably rescaling the remaining elements.
For instance, the 3-inflation of the permutation $451\underline{3}2$ at 3 is $671\underline{345}2$.

The element $\pi_i$ of $\pi$ is called a \emph{left-to-right maximum} (briefly, LTR maximum)
when it is bigger than all the elements preceding it, i.e. $\pi_i >\max (\pi_1 ,\ldots ,\pi_{i-1})$.
The permutation $\underline{3}1\underline{5}\underline{7}624\underline{9}8$ has four LTR maxima,
which are the elements underlined.

The usual symmetries of a permutation are the \emph{reverse, inverse and complement} operations.
Given $\pi =\pi_1 \pi_2 \cdots \pi_n$, we define its reverse $\pi^r =\pi_n \pi_{n-1}\cdots \pi_1$,
its complement $\pi^c =(n+1-\pi_1 )(n+1-\pi_2 )\cdots (n+1-\pi_n )$
and its inverse $\pi^{-1}$ as the usual group-theoretic inverse.

A \emph{Dyck path} is a path in the discrete plane $\mathbb{Z}\times \mathbb{Z}$
starting at the origin of a fixed Cartesian coordinate system, ending on the $x$-axis,
never falling below the $x$-axis and using two kinds of steps (of length 1),
namely up steps $\U=(1,1)$ and down steps $\D=(1,-1)$.
The \emph{length} of a Dyck path is its final abscissa, which coincides with the total number of its steps.
For instance, $\U\U\D\U\U\D\D\D\U\D$ is a Dyck path of length 10.
According to their semilength, Dyck paths are counted by Catalan numbers (sequence A000108 in \cite{Sl}).
The $n$-th \emph{Catalan number} is $C_n =\frac{1}{n+1}{2n\choose n}$
and the associated generating function is $C(x)=(1-\sqrt{1-4x})/(2x)$.

There is a well known bijection between 213-avoiding permutations (of length $k$) and Dyck paths (of semilength $k$), which can be succinctly described as follows:
given a Dyck path $P$ of semilength $k$,
label its down steps from right to left with positive integers 1 to $k$ in increasing way,
then label each up step with the same label as the down step it is matched with,
finally read the labels of the up steps from left to right, so to obtain a 213-avoiding permutation.
For instance, the above Dyck path $\U\U\D\U\U\D\D\D\U\D$ corresponds to the permutation 25341,
which in fact avoids 213 (see Figure~\ref{Dyck}).
See \cite{Kr} for an equivalent version of the above bijection using 132-avoiding permutations.

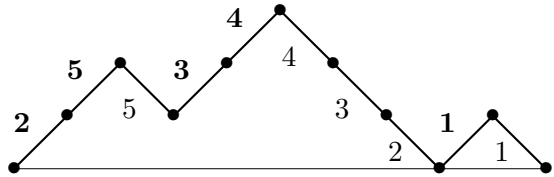
\begin{figure}[h!]
\begin{center}
\begin{tikzpicture}[scale=0.7]
\draw [ultra thin] (0,0) -- (10,0);
\draw [thick] (0,0) -- (2,2);
\draw [thick] (2,2) -- (3,1);
\draw [thick] (3,1) -- (5,3);
\draw [thick] (5,3) -- (8,0);
\draw [thick] (8,0) -- (9,1);
\draw [thick] (9,1) -- (10,0);
\node[] at (0,0) {$\bullet$};
\node[] at (1,1) {$\bullet$};
\node[] at (2,2) {$\bullet$};
\node[] at (3,1) {$\bullet$};
\node[] at (4,2) {$\bullet$};
\node[] at (5,3) {$\bullet$};
\node[] at (6,2) {$\bullet$};
\node[] at (7,1) {$\bullet$};
\node[] at (8,0) {$\bullet$};
\node[] at (9,1) {$\bullet$};
\node[] at (10,0) {$\bullet$};
\node [above left] at (0.5,0.5) {$\mathbf{2}$};
\node [above left] at (1.5,1.5) {$\mathbf{5}$};
\node [above left] at (3.5,1.5) {$\mathbf{3}$};
\node [above left] at (4.5,2.5) {$\mathbf{4}$};
\node [above left] at (8.5,0.5) {$\mathbf{1}$};
\node [below left] at (2.5,1.5) {$5$};
\node [below left] at (5.5,2.5) {$4$};
\node [below left] at (6.5,1.5) {$3$};
\node [below left] at (7.5,0.7) {$2$};
\node [below left] at (9.5,0.7) {$1$};
\end{tikzpicture}
\caption{A Dyck path and its associated 213-avoiding permutation
(read the bold labels from left to right).}
\label{Dyck}
\end{center}
\end{figure}

Another classical family of lattice paths is that of Schr\"oder paths.
A \emph{Schr\"oder path} is defined exactly like a Dyck path, except that it uses one more kind of steps,
namely double horizontal steps $\H_2 =(2,0)$.
The length of a Schr\"oder path does not coincide with the number of its steps:
it can be rather obtained as the sum of the number of its up steps and down steps
with twice the number of its double horizontal steps.
As an example, $\U\H_2 \U\D\D\H_2 \U\D$ is a Schr\"oder path of length 10
.

Our main goal is to study the sorting power of $\sigma$-machines.
We remark that, due to the specificity of our setting,
a permutation $\pi$ is $\sigma$-sortable if and only if
the output $s_{\sigma}(\pi )$ resulting from the first pass (through the $\sigma$-avoiding stack) avoids 231.
This fact (as well as the notation $s_{\sigma}(\pi )$) will be frequently used throughout the paper,
especially in Section~\ref{123machine}.

We close this section by analyzing the $\sigma$-machine when $\sigma$
has length 2.  If $\sigma =21$, the $\sigma$-machine is precisely West's
right greedy algorithm performed on two stacks in series, so we refer
 to \cite{W}.  If $\sigma =12$, we have the following result,
which completely characterize and enumerate 12-sortable permutations.
We remark that the $12$-machine is different from the one considered in
\cite{Sm}, which is constituted by a
decreasing stack and an increasing stack connected in series.  Smith's
device does not require that the input permutation makes a
complete pass through the first stack before entering the second one.

\begin{theorem}\label{12machine}
  A permutation $\pi$ is 12-sortable if and only if $\pi \in \Av(213)$.
  As a consequence, $|\Sort_n(12)|=C_n$, the $n$-th Catalan number.
\end{theorem}

\begin{proof}
  Write $\pi$ as $\pi =L1R$, where $L$ is the prefix of $\pi$ preceding
  1 and $R$ is the suffix of $\pi$ following 1.  It is easy to see that
  $s_{12}(\pi )=s_{12}(L1R)=s_{12}(L)s_{12}(R)1$.  This is because 1
  enters the stack only when the stack itself is empty, and exits the
  stack only at the end.  Now, using induction on the length of $\pi$
  and a simple case by case analysis, it is not difficult to show that,
  if $\pi$ contains 213, then $s_{12}(\pi )$ contains 231, and so $\pi$
  is not sortable.

  On the other hand, suppose that $\pi$ is not sortable, and so that
  $s_{12}(\pi )$ contains 231.  We wish to show that $\pi$
  contains 213.  Indeed, if $bca$ is an occurrence of 231 in
  $s_{12}(\pi )$ (with $a<b<c$), then necessarily $b$ comes before $c$
  also in $\pi$ (since a non-inversion in the output necessarily comes
  from a non-inversion in the input).  However, $b$ exits the stack
  before $c$ enters it, and this must be due to the presence of an
  element $x$, located between $b$ and $c$ in $\pi$, which is smaller
  than $b$. The three elements $b,x$ and $c$ are thus an occurrence of
  213 in $\pi$.
\end{proof}

The above results implies that $\Sort (21)$ is not a class, whereas $\Sort (12)$ is.


\section{Classes and nonclasses of $\sigma$-sortable permutations}\label{classes}

Given a permutation $\sigma$, it is natural to ask if $\sigma$-sortable permutations constitute a permutation class.
Concerning permutations of length 2, as we have already observed, 12-sortable permutations are a class, whereas 21-sortable permutations are not.
Concerning permutations of length 3, some computations suggest that 321-sortable permutations might be a class,
whereas in the five remaining cases $\sigma$-sortable permutations are not a class, as it can be deduced from the table below:

\begin{center}
\begin{tabular}{|c|c|c|}
  \hline
  $\sigma$ & $\sigma$-sortable permutation & non-$\sigma$-sortable pattern \\
  \hline
  123 & 4132 & 132 \\
  \hline
  132 & 2413 & 132 \\
  \hline
  213 & 4132 & 132 \\
  \hline
  231 & 361425 & 1324 \\
  \hline
  312 & 3142 & 132 \\
  \hline
\end{tabular}
\end{center}

Looking at more data, by taking longer permutations, suggests a rather surprising conjecture:
the number of permutations $\sigma$ such that $\sigma$-sortable permutations are not a class is the $n$-th Catalan number!
The rest of this section is devoted mainly to prove this conjecture, as well as to provide some related results.

\bigskip

We start by showing how the sortability of a permutation is affected by its connections with the pattern defining the constraint of the stack.

\begin{lemma}
Let $\sigma=\sigma_1 \sigma_2 \cdots \sigma_k$. Suppose to perform the $\sigma$-machine on the permutation $\pi$.
\begin{enumerate}

\item If $\pi \in \Av (\sigma^r )$, then the output of the first stack is $\pi^r$.
Therefore $\pi$ is $\sigma$-sortable if and only if $\pi \in \Av (132)$.

\item If $\pi$ contains the pattern $\sigma^r$, then the output of the first stack contains the pattern $\hat{\sigma}=\sigma_2 \sigma_1 \sigma_3 \sigma_4 \cdots \sigma_k$, obtained by interchanging the first two elements of $\sigma$. Therefore, if $\hat{\sigma}$ contains the pattern $231$, then $\pi$ is not $\sigma$-sortable.

\end{enumerate}
\end{lemma}

\begin{proof}
\begin{enumerate}
\item If $\pi \in \Av (\sigma^r )$, then the restriction of the first stack never triggers, so it outputs $\pi^r$.
Therefore $\pi$ is sortable if and only if it avoids $132$.
\item Suppose that $\pi$ contains $\sigma^r$. Let $s_k s_{k-1} \cdots s_1$ be the (lexicographically) leftmost occurrence of $\sigma^r$ in $\pi$.
Then every element of $\pi$ is pushed into the first stack until $s_1$ is scanned.
Before pushing $s_1$ into the stack, the element $s_2$ has to be popped out due to the $\sigma$-restriction.
Moreover, the element $s_3$ is not popped before $s_1$ enters the stack,
otherwise in $\pi$ there would be an occurrence of $\sigma^r$ to the left of $s_k s_{k-1}\cdots s_1$, which is a contradiction.
Therefore $s_1$ is pushed into the first stack when $s_3$ is still inside:
this is enough to conclude that the output of the first stack will contain $s_2 s_1 s_3 \cdots s_k$, which is an occurrence of $\hat{\sigma}$.
Therefore, if $\hat{\sigma}$ contains $231$, then $\pi$ is not $\sigma$-sortable.
\end{enumerate}
\end{proof}

\begin{theorem}\label{sufficient}
Let $\sigma=\sigma_1 \sigma_2 \cdots \sigma_k$ and $\hat{\sigma}$ as in the above lemma.
If $\hat{\sigma}$ contains the pattern $231$, then $\Sort (\sigma)=\Av (132,\sigma^r)$.
In such a case, $\Sort (\sigma)$ is a class with basis either $\{ 132,\sigma^r \}$ (if $\sigma^r \in \Av (132)$) or $\{ 132\}$ (otherwise).
\end{theorem}

\begin{proof}
Given any permutation $\pi$, suppose that $\pi$ contains $\sigma^r$.
Then the previous lemma implies that the output of the first stack contains $\hat{\sigma}$.
Since $\hat{\sigma}$ contains $231$ by hypothesis, $\pi$ is not $\sigma$-sortable.
Now suppose that $\pi$ avoids $\sigma^r$, but that it contains $132$.
Again as a consequence of the previous lemma, the output of the first stack is then $\pi^r$, and $\pi^r$ contains $231$,
so $\pi$ is not $\sigma$-sortable.
We have thus proved that $\Sort (\sigma )\subseteq \Av (132,\sigma^r )$.

Conversely, suppose that $\pi$ avoids both $132$ and $\sigma^r$.
Then the previous lemma implies that the output of the first stack is $\pi^r$, which avoids $132^r=231$ by hypothesis,
therefore $\pi$ is $\sigma$-sortable.
Thus we also have that $\Av (132,\sigma^r )\subseteq \Sort (\sigma )$, and so the equality holds.
\end{proof}

\begin{corollary}\label{321}
For all $k\geq 3$, $\Sort (k(k-1)\cdots 1)=\Av (132,12\cdots k)$.
In particular, the set of 321-sortable permutations is a class.
\end{corollary}

Theorem~\ref{sufficient} provides a sufficient condition for a permutation $\sigma$ in order to have that $\Sort (\sigma )$ is a class.
It turns out that this condition is also necessary.

\begin{theorem}
If $\hat{\sigma}$ avoids the pattern $231$, then $\Sort (\sigma )$ is not a permutation class.
\end{theorem}

\begin{proof}
The above corollary and the table at the beginning of this section tells that, if $\sigma$ has length at most $3$,
the theorem holds.

Now suppose that $\sigma$ has length at least $4$.
It is not hard to realize that the permutation $132$ is not $\sigma$-sortable, since the output of the first stack is $231$.
We now show that, if $\hat{\sigma}$ avoids $231$,
then it is always possible to construct a permutation $\alpha$ such that $\alpha$ contains $132$ and $\alpha$ is $\sigma$-sortable,
thus proving that $\Sort(\sigma )$ is not a class.
Suppose, as usual, that $\sigma=\sigma_1 \sigma_2 \cdots \sigma_k$.
We distinguish two cases, depending on the relative order of the elements $\sigma_1$ and $\sigma_2$.

\begin{enumerate}

\item If $\sigma_1 < \sigma_2$, define $\alpha=\sigma'_k \sigma'_{k-1} \cdots \sigma'_3 \ z \ \sigma'_2 \sigma'_1$, where
\begin{itemize}
\item $z=\sigma_1$,
\item $\sigma'_i=
\begin{cases}
\sigma_i, & \mbox{ if } \sigma_i < \sigma_1; \\
\sigma_i +1, & \mbox{ otherwise.}
\end{cases}$
\end{itemize}

Note that $z \sigma'_2 \sigma'_1$ is an occurrence of $132$.
We will show that $\alpha$ is $\sigma$-sortable by means of a detailed analysis of the behavior of the $\sigma$-machine on input $\alpha$.
The elements of $\alpha$ are pushed into the first stack until $\sigma'_1$ is scanned
(it is the first elements that triggers the restriction $\Av(\sigma )$).
In particular, both the additional element $z$ and $\sigma'_2$ can be pushed into the stack,
because $\sigma'_2 z \cdots \sigma'_{k-1} \sigma'_k$ is not an occurrence of $\sigma$ (since $\sigma_1 < \sigma_2$ and $\sigma'_2 >z$).
Now, before $\sigma'_1$ enters the first stack, the element $\sigma'_2$ is extracted and pushed into the second stack.
At this point, $\sigma'_1$ can enter without violating the restriction, again because $\sigma_2 > \sigma_1$, whereas $z < \sigma'_1$,
so that $\sigma'_1 z \sigma'_3 \cdots \sigma'_k$ is not an occurrence of $\sigma$.
As a result, the output of the first stack is $out=\sigma'_2 \sigma'_1 z \sigma'_3 \cdots \sigma'_k$,
so it will be enough to show that $out$ does not contain an occurrence of the pattern $231$.
Since $\hat{\sigma}$ avoids $231$ by hypothesis, a potential occurrence of $231$ necessarily involves the new element $z$.
In particular, it is easy to observe that $z$ can be neither the smallest nor the biggest element of such a pattern,
because $z < \sigma'_1 < \sigma'_2$ and $z$ is the third element of $out$.
Finally, if $z$ were the first element of an occurrence $z \ \sigma'_j \ \sigma'_l$ of $231$ in $out$,
then $\sigma_1 \ \sigma_j \ \sigma_l$ would be an occurrence of $231$ in $\hat{\sigma}$, against the hypothesis.

\item If $\sigma_1 > \sigma_2$, define $\alpha=\sigma'_k \sigma'_{k-1} \cdots \sigma'_3  \sigma'_2 \sigma'_1 \ z$, where
\begin{itemize}
\item $z=\sigma_2 +1$.
\item $\sigma'_i=
\begin{cases}
\sigma_i ,   & \mbox{ if } \sigma_i \le \sigma_2 ; \\
\sigma_i+1 , & \mbox{ otherwise.}
\end{cases}$
\end{itemize}

Observe that $\sigma'_2 \sigma'_1 z$ is an occurrence of $132$.
As for the previous case, we now describe what happens when $\alpha$ is processed by the $\sigma$-machine.
The first element that cannot be pushed into the first stack is $\sigma'_1$, which forces $\sigma'_2$ to be extracted.
Successively both $\sigma'_1$ and $z$ can enter the first stack, since $z \sigma'_1 \sigma'_3 \cdots \sigma'_k$ is not an occurrence of $\sigma$
(indeed $\sigma_1 > \sigma_2$ and $z < \sigma'_1$).
Therefore the output of the first stack is $out = \sigma'_2 z \sigma'_1 \sigma'_3 \cdots \sigma'_k$,
and again a potential occurrence of $231$ in $out$ must involve the new element $z$.
However $z$ cannot be the smallest element of a pattern $231$, because it is the second element of $out$.
Moreover, if $z$ were the first element of a $231$, then $\sigma_2$ would be the first element of an occurrence of $231$ in $\hat{\sigma}$,
which is forbidden.
Finally, if $z$ were the largest element of a $231$, then $\sigma'_2$ would be the first element of such an occurrence,
so also $\sigma'_1$, which is greater than both $\sigma'_2$ and $z$, would be the largest element of an occurrence of $231$ which does not involve $z$,
giving again a contradiction.
Thus we have showed that $out$ does not contain the pattern $231$, which means that $\alpha$ is $\sigma$-sortable.
\end{enumerate}
\end{proof}

\begin{corollary}\label{characterization}
For every permutation $\sigma$, the set $\Sort(\sigma )$ of the permutations sortable using the $\sigma$-machine is not a permutation class if and only if $\hat{\sigma}$ avoids the pattern $231$.
\end{corollary}

\begin{corollary}\label{enum}
The permutations $\sigma$ for which $\Sort(\sigma )$ is not a permutation class are enumerated by Catalan numbers.
\end{corollary}
\begin{proof}
Such permutations are in bijection with $\Av(231)$, which is known to be enumerated by Catalan numbers.
\end{proof}

We have thus shown that $\Sort(\sigma )$ is a permutation class if and only if $\hat{\sigma}$ contains the pattern $231$.
In this case,  $\Sort(\sigma )=\Av(132,\sigma^r )$, hence the basis of $\Sort(\sigma )$ has exactly two elements if and only if $\sigma$ avoids $231$.
We next give exact enumerative results concerning $\Sort(\sigma )$ when its basis has two elements.

\begin{prop}
Suppose that $\sigma^r$ avoids the pattern $132$.
Then $\hat{\sigma}$ contains the pattern $231$ if and only if $\sigma_1 \sigma_2 \sigma_3$ is an occurrence of the pattern $321$.
\end{prop}

\begin{proof}
By hypothesis, $\sigma$ avoids $231$.
If $\hat{\sigma}=\sigma_2 \sigma_1 \sigma_3 \cdots \sigma_k$ contains the pattern $231$,
then there can be only one occurrence of $231$ and it has to involve both $\sigma_1$ and $\sigma_2$,
respectively as the first and the second element of the pattern, with $\sigma_2 < \sigma_1$.
Let $\sigma_2 \sigma_1 \sigma_i$ be such an occurrence, with $i \ge 3$.
If $\sigma_3 > \sigma_2$, then $i>4$, and $\sigma_2 \sigma_3 \sigma_i$ would be an occurrence of $231$ in $\sigma$, which is impossible.
Therefore it must be $\sigma_3 < \sigma_2$, hence $\sigma_1 > \sigma_2 > \sigma_3$, as desired.

Conversely, if $\sigma_1 \sigma_2 \sigma_3$ is an occurrence of the pattern $321$,
then clearly $\sigma_2 \sigma_1 \sigma_3$ is an occurrence of $231$ in $\hat{\sigma}$.
\end{proof}

\begin{prop}
Let $a_n= | \left\lbrace \pi \in \Av_n (231)\, |\, \pi_1 \pi_2 \pi_3 \simeq 321 \right\rbrace |$; then, for each $n \ge 2$, we have $a_n=C_n -2C_{n-1}$. In particular, the generating function of the sequence $(a_n )_{n \ge 0}$ is $$A(x)=\frac{1-4x+2x^2 -(1-2x)\sqrt{1-4x}}{2x}.$$
\end{prop}

\begin{proof}
Suppose that $n \ge 2$.
It is well known that $|\Av_n (231)|=C_n$, hence we have $a_n=C_n -(f_n +g_n )$, where
$$
\begin{cases}
\mathcal{F}_n= \left\lbrace \pi \in \Av_n (231)\, |\, \pi_1 < \pi_2 \right\rbrace, \hfill f_n=|\mathcal{F}_n|; \\
\mathcal{G}_n= \left\lbrace \pi \in \Av_n (231)\, |\, \pi_1 > \pi_2, \ \pi_2 < \pi_3 \right\rbrace, \ \ \ \hfill g_n=|\mathcal{G}_n|.
\end{cases}
$$

We now show that $f_n=g_n=C_{n-1}$ by explicitly finding bijections between $\mathcal{F}_n$ and $\Av_{n-1}(231)$ as well as between $\mathcal{G}_n$ and $\Av_{n-1}(231)$,
thus obtaining the desired enumeration.
\begin{itemize}
\item If $\pi \in \mathcal{F}_n$, then it must be $\pi_1=1$, otherwise $\pi_1 \pi_2 1$ would be an occurrence of $231$ in $\pi$.
    Thus we can define $f:\mathcal{F}_n \rightarrow \Av_{n-1}(231)$ such that
    $f(\pi)$ is obtained from $\pi$ by removing $\pi_1=1$ and rescaling the remaining elements.
    It is clear that $f(\pi) \in \Av_n (231)$ and that $f$ is an injection.
    Moreover, if $\tau \in \Av_{n-1}(231)$, then adding a new minimum at the beginning cannot create any occurrence of $231$,
    so $f$ is also surjective.

\item If $\pi \in \mathcal{G}_n$, then it must be $\pi_2=1$, otherwise the elements $\pi_2 \pi_3 1$ would form an occurrence of $231$ in $\pi$.
We thus define $g:\mathcal{F}_n \rightarrow \Av_{n-1}(231)$ such that
$g(\pi )$ is obtained from $\pi$ by removing $\pi_2 =1$ and rescaling the remaining elements.
Again it is clear that $g(\pi ) \in \Av_n (231)$ and that $g$ is an injection.
Finally, if $\tau \in \Av_{n-1}(231)$, then the permutation $\pi$ obtained from $\tau$ by adding a new minimum in the second position avoids $231$,
because a potential occurrence of $231$ in $\pi$ should involve the added element $\pi_2$,
and so $\pi_2$ would be either the first or the second element of such an occurrence, which cannot be since $\pi_2=1$.
Therefore $g$ is a bijection between $\mathcal{F}_n$ and $\Av_{n-1}(231)$, as desired.

\end{itemize}
We can now compute the generating function $A(x)=\sum_{n \ge 0} a_n x^n$ as follows:
\begin{equation*}
\begin{split}
A(x)&=\sum_{n \ge 0}a_{n+2}x^{n+2}=\sum_{n \ge 0}C_{n+2}x^{n+2}-2x \sum_{n \ge 0}a_{n+1}x^{n+1} \\
&=C(x)-x-1 -2x(C(x)-1)=C(x)(1-2x)+x-1,
\end{split}
\end{equation*}
from which we get $A(x)=(1-4x+2x^2 -(1-2x)\sqrt{1-4x})/(2x)$, as desired.
\end{proof}

Sequence $(a_n )_{n \ge 0}$ starts $0,0,1,4,14,48,165,572,2002,\dots$ and is recorded as sequence A002057 in \cite{Sl} (with offset 2).
Observe that $A(x)=x^2 C(x)^4$, a fact for which we do not have a combinatorial explanation.

\bigskip

In Figure~\ref{table} we report some enumerative results concerning classes of $\sigma$-sortable permutation with basis of cardinality 2.
Each case can be proved with a direct combinatorial argument.

\begin{figure}
\begin{center}
\begin{tabular}{|c|c|c|c|c|}
  \hline
  length & Pattern $\sigma$ & G.F. & Sequence & OEIS \\
  \hline
  \hline
  \textbf{3} & 321 &\vphantom{\Big|} $\frac{1-x}{1-2x}$ & 1,1,2,4,8,16,32,64,128,256,512,$\dots$ & A000079 \\
  \hline
  \hline
  \textbf{4} & 3214 & & & \\
  & 4213 & $\frac{1-2x}{1-3x+x^2}$ & 1,1,2,5,13,34,89,233,610,1597,4181,$\dots$ & A001519 \\
  & 4312 & & & \\
  & 4321 & & & \\
  \hline
  \hline
  \textbf{5} & 32145 &\vphantom{\bigg|} $\frac{-3x^4+9x^3-12x^2+6x-1}{(x-1)(x^2-3x+1)^2}$ & 1,2,5,14,41,121,355,1032,2973,8496,$\dots$ & A116845 \\
  \hline
  & 52134 &\vphantom{\bigg|} $\frac{(1-x)(2x-1)^2}{x^4-9x^3+12x^2-6x+1}$ & 1,2,5,14,41,121,355,1033,2986,8594,$\dots$ & not in \cite{Sl} \\
  \hline
  & 54123 &\vphantom{\bigg|} $\frac{1-4x+5x^2-3x^3}{x^4 -6x^3+8x^2-5x+1}$ & 1,2,5,14,41,121,356,1044,3057,8948, $\dots$ & not in \cite{Sl} \\
  \hline
  & 32154 & & & \\
  & 42135 & & & \\
  & 43125 & & & \\
  & 43215 & & & \\
  & 52143 & & & \\
  & 53124 & $\frac{x^2-3x+1}{3x^2-4x+1}$ & 1,2,5,14,41,122,365,1094,3281,9842,$\dots$ & A124302 \\
  & 53214 & & & \\
  & 54132 & & & \\
  & 54213 & & & \\
  & 54312 & & & \\
  & 54321 & & & \\
  \hline
\end{tabular}
\end{center}
\caption{Classes of $\sigma$-sortable permutations whose basis has two elements.}
\label{table}
\end{figure}

\section{The 321-machine}\label{321machine}

When the first stack is $\pat{3\\2\\1}$-avoiding,
the results of the previous section, specifically Corollary~\ref{321},
tell that 321-sortable permutations constitute a class
(as a matter of fact, this is the only class of $\sigma$-sortable permutations for $\sigma$ of length 3).
We in fact have the more general result that,
if we set $\rho_k=k(k-1)\cdots 21$ (i.e. $\rho_k$ is the reverse identity permutation of length $k$),
then $\Sort(\rho_k)=\Av(12 \cdots k,132)$.

For small values of $k$, we have the following table, where the row
labelled $k$ records the number of permutations of length $n$ sortable
by the $\rho_k$-machine:

{\scriptsize
  \begin{center}
    \bgroup
    \def\arraystretch{1.5}
    \begin{tabular}{c|rrrrrrrrrrrr|c}
      $k\setminus n$ & 0 & 1 & 2 & 3 & 4 & 5 & 6 & 7 & 8 & 9 & 10 & 11 & OEIS \\
      \hline
      3 & 1 & 1 & 2 & 4 & 8 & 16 & 32 & 64 & 128 & 256 & 512 & 1024 & A011782 \\
      4 & 1 & 1 & 2 & 5 & 13 & 34 & 89 & 233 & 610 & 1597 & 4181 & 10946 & A001519 \\
      5 & 1 & 1 & 2 & 5 & 14 & 41 & 122 & 365 & 1094 & 3281 & 9842 & 29525 & A124302 \\
      6 & 1 & 1 & 2 & 5 & 14 & 42 & 131 & 417 & 1341 & 4334 & 14041 & 45542 & A080937 \\
      7 & 1 & 1 & 2 & 5 & 14 & 42 & 132 & 428 & 1416 & 4744 & 16016 & 54320 & A024175 \\
    \end{tabular}
    \egroup
  \end{center}
}

Since $|\Av_n(132)|=C_n$, it is clear that the rows tend to the sequence of Catalan numbers.
For $k=3$, we have that $|\Sort_n (321)|=2^{n-1}$ (for $n\geq 1$); this is sequence A011782 in the OEIS~\cite{Sl}.
Looking at the OEIS references (reported in the above table), we observe that,
for any given $k$, the associated sequence counts the number of Dyck paths of height at most $k-1$.
This can be proved by
using the mentioned bijection between Dyck paths and 132-avoiding permutations described in \cite{Kr},
observing that the maximum length of an increasing sequence corresponds to the height of the path.
Dyck paths of bounded height are rather well studied objects, see for example \cite{BM,GX}.

Exploiting this connection,
we can give a description of the generating function of the sequence recorded in the $k$-th row.
Using the usual ``first-return'' decomposition of Dyck paths,
it is possible to find a recursive description of the generating function
$F_k (x)$ of Dyck paths of height at most $k$ with respect to the semilength:
$F_0 (x)=1$ and, for $k\geq 1$,
$$
F_k (x)=1+xF_{k-1}(x)F_k (x).
$$

From the above recurrence it is immediate to see that $F_k (x)$ is rational, for all $k$;
more specifically, we have $F_k (x)=G_k (x)/G_{k+1}(x)$,
where $G_k(x)$ satisfies the recurrence $G_{k+1}(x)=G_k (x)-xG_{k-1}(x)$,
with initial conditions $G_0 (x)=G_1 (x)=1$.
Solving this recurrence yields $G_k (x)=\sum_{i\geq 0}{n-i\choose i}(-x)^i$.
The polynomials $G_k (x)$ are sometimes called \emph{Catalan polynomials}, see for instance \cite{CLF};
the table of their coefficients is sequence A115139 in \cite{Sl}.

\section{The 123-machine}\label{123machine}

Now suppose that the first stack is $\pat{1\\2\\3}$-avoiding.
%
%
This machine is considerably more challenging that the previous one.
The first thing we observe is that, unlike the previous case,
the set of sortable permutation is not a class, as we already knew.
For instance, the permutation 4132 is sortable, whereas its pattern 132 is not.

First, we show that a necessary condition for a permutation to be
sortable is that the first two elements are not a large ascent.

\begin{lemma}
  If $\pi =\pi_1 \pi_2 \cdots \pi_n$ is sortable, then
  $\pi_2 \leq \pi_1 +1$.
\end{lemma}

\begin{proof}
  Suppose that $\pi_2 \geq \pi_1 +2$.  This implies that there exists an
  index $i\in \{ 3,4,\dots ,n\}$ such that $\pi_2 >\pi_i >\pi_1$.
  During the sorting process, the first two elements $\pi_1$ and $\pi_2$
  enters the stack and reach the output only when the stack is emptied
  at the end of the process (since $\pi_1 <\pi_2$).  Thus the three
  elements $\pi_i ,\pi_2$ and $\pi_1$ show an occurrence of 231 in
  $s_{123}(\pi )$, hence $\pi$ is not sortable.
\end{proof}

As a consequence of the above lemma, we can partition the set of
sortable permutations into two classes: those starting with an
ascent consisting of two consecutive values and those starting with a descent.  We now wish to
show that, given a sortable permutation, we can add in front of it an
arbitrary number of consecutive ascents and the resulting permutation is
still sortable.

\begin{lemma}
  Let $\pi= \pi_1 \pi_2 \cdots \pi_n$ and let $\pi'$ the permutation (of
  length $n+1$) obtained from $\pi$ by 2-inflating its first element
  $\pi_1$.  Then $\pi$ is sortable if and only if $\pi'$ is sortable.
\end{lemma}

\begin{proof}
  Observe that, by hypothesis, the first two elements of $\pi'$ are
  consecutive in value ($x$ and $x+1$, say) and the first one is smaller
  than the second one.  Therefore, during the sorting process, such two
  elements remain at the bottom of the stack (with $x+1$ above $x$)
  until all the other elements of the input permutations have exited the
  stack.  Moreover, the behavior of the stack is not affected by the
  presence of $x+1$, meaning that $x$ and $x+1$ can be considered as a
  single element.  As a consequence, the last
  two elements of $s_{123}(\pi')$ are $x+1$ and $x$, and then that
  $s_{123}(\pi )$ contains 231 if and only $s_{123}(\pi')$ contains 231.
\end{proof}

\begin{corollary}\label{inflation}
  Given $\pi \in S_n$, let $\pi'$ be obtained from $\pi$ by
  $k$-inflating the first element of $\pi$ (with $k\geq 1$).  Then $\pi$
  is sortable if and only if $\pi'$ is sortable.
\end{corollary}


The above results tell us that, up to ``deflating'' the prefix of
consecutive ascents (if there is one), we can restrict to investigate
sortability of permutations starting with a descent.  Denote by
$\DSort_n(123)$ this subset of $\Sort_n(123)$; that is,
$\DSort_n(123) =\{ \pi \in \Sort_n(123) \, |\, \pi_1 >\pi_2 \}$.  Our goal is now
to characterize and enumerate $\DSort_n(123)$.

\begin{lemma}\label{starting_descent}
  Let $\pi \in \DSort_n(123)$, with $\pi_1 =k$. Then we have
  $s_{123}(\pi )=n(n-1)\cdots (k+1)(k-1)\cdots 21k$.
\end{lemma}

\begin{proof}
  Let $\gamma=s_{123}(\pi )=\gamma_1 \gamma_2 \cdots \gamma_n$. Clearly
  $\gamma_n =k$.  Now suppose that the two elements $x$ and $y$
  constitute an ascent in $\gamma$, with $x<y$ and $y\neq k$.

  We first show that $y$ comes before $x$ in $\pi$.  Suppose in fact
  that this is not the case, and focus on the instant when $x$ exits the
  first stack.  We distinguish two cases.

  \begin{itemize}
  \item $x$ exits the first stack because it is the second element of a
    pattern 321 in $\pi$.  More specifically, there is an element $c$ in
    the stack such that $c>x$ and the next element $a$ of $\pi$ is such
    that $x>a$.  This implies, in particular, that $a\neq y$, and so
    that $y$ follows $a$ in $\pi$.  Therefore $s_{123}(\pi )$ contains
    either the subword $xay$, which is impossible since $x$ and $y$ are
    supposed to be consecutive in $s_{123}(\pi )$, or the subword $xya$,
    which is impossible too since otherwise $s_{123}(\pi )$ would
    contain the pattern 231.
  \item $x$ exits the first stack because the next element $a$ of $\pi$
    is smaller than two elements $b$ and $c$ in the stack, with $b<c$
    and $c$ deeper than $b$.  In this case, $s_{123}(\pi )$ would
    contain the subword $xby$, which is impossible, again because $x$
    and $y$ would not be consecutive.
  \end{itemize}

  Thus we can write $\pi$ as $\pi=k\pi_2 \cdots y\cdots x\cdots$.  Since
  $x$ and $y$ are consecutive in $s_{123}(\pi )$, $x$ must enter the
  stack just above $y$.  This implies, in particular, that
  $y\geq \pi_1$, otherwise $\pi_1 ,y$ and $x$ would constitute a
  forbidden pattern inside the stack.

  We also notice that, when $x$ enters the first stack, at the bottom of
  the stack there is at least one element $w<\pi_1$ just above $\pi_1$.
  Indeed, either $\pi_2$ is still in the stack (and in this case
  $w=\pi_2$) or $\pi_2$ has been forced to exit the stack by an element
  $\tilde{w}<\pi_2 <\pi_1$; in this case, $\tilde{w}$ replaces $\pi_2$
  just above $\pi_1$.  Iterating this argument, we get the desired
  property.

  Summing up, when $x$ enters the first stack, the stack itself contains
  the elements (from bottom to top) $\pi_1 ,w,y,x$.  Now, we have that
  $x>w$, otherwise $\pi_1 ,w$ and $x$ would constitute a forbidden
  pattern in the stack.  Hence $s_{123}(\pi )$ must contain the subword
  $xyw$, which is isomorphic to the pattern 231; this means that $\pi$
  is not sortable.
\end{proof}

\begin{corollary}\label{before_maximum}
  Let $\pi \in \DSort_n(123)$ and suppose that $\pi_1 \neq n$.
  Also, suppose that $\pi_i =n$.  Then either
  $\pi_{i-1}=n-1$ (if $\pi_1 \neq n-1$) or $\pi_{i-1}=n-2$ (if
  $\pi_1 =n-1$).
\end{corollary}

\begin{proof}
  Notice that $i\in \{ 3,4,\ldots, n\}$ (indeed $i\neq 1$ by hypothesis
  and $i\neq 2$ since $\pi$ starts with a descent).  The element $n$
  enters the first stack immediately above $\pi_{i-1}$, since pushing
  $n$ into the stack can never generate a forbidden pattern.  Moreover,
  $n$ and $\pi_{i-1}$ exit the stack together, since $n$ cannot play the
  role of the second element in a forbidden pattern inside the stack.
  Therefore, $s_{123}(\pi )$ contains the factor $n\pi_{i-1}$.
  By Lemma~\ref{starting_descent}, this implies the result.
\end{proof}

\begin{corollary}\label{starting_maximum}
  The set of permutations of $\DSort_n(123)$ starting with $n$ is the set of
  213-avoiding permutations of length $n$ starting with $n$, for all
  $n\geq 2$.
\end{corollary}

\begin{proof}
  Let $\pi \in \DSort_n(123)$, and suppose that $\pi$ starts with $n$.
  As soon as $n$ enters the stack, it makes the stack act as a
  $\pat{1\\2}$-avoiding stack for the rest of the permutation.
  Therefore, by Theorem~\ref{12machine}, $\pi$ is
  sortable if and only if the permutation obtained from $\pi$ by
  removing the first element avoids 213, which is in turn equivalent to
  the fact that $\pi$ avoids 213.
\end{proof}

Since it is well known that 213-avoiding permutations are counted by
Catalan numbers, an immediate consequence of the previous corollary is
that the number of permutations of $\DSort_n(123)$ starting with $n$ is the
$(n-1)$-th Catalan number $C_{n-1}$.

In order to completely characterize the set $\DSort_n(123)$, what we need to
do is to analyze the subset of $\DSort_n(123)$ consisting of permutations
which do not start with $n$.  In other words, these are the permutations
of $\DSort_n(123)$ having at least two LTR maxima.  Denote this set by
$\DSort_n({\geq}2; 123)$.  Moreover, the set of permutations of $\DSort_n(123)$
having precisely $i$ LTR maxima will be denoted $\DSort_n(i;123)$.

\begin{theorem}\label{bijection}
  Let $n\geq 3$. There exists a bijection
  $\varphi :\DSort_{n-1}(123)\rightarrow \DSort_n({\geq} 2,123)$.  Moreover, the
  restriction of $\varphi$ to $\DSort_{n-1}(i;123)$ is a bijection between
  $\DSort_{n-1}(i;123)$ and $\DSort_n(i+1;123)$.
\end{theorem}

\begin{proof}
  Let $\pi =\pi_1 \cdots \pi_{n-1} \in \DSort_{n-1}(123)$.  Let
  $\varphi (\pi )=\tilde{\pi}$ be obtained from $\pi$ by inserting $n$:
  \begin{itemize}
  \item either after $n-1$, if $\pi_1 \neq n-1$, or
  \item after $n-2$, if $\pi_1 =n-1$.
  \end{itemize}

  First we show that $\varphi$ is well defined, i.e. that
  $\tilde{\pi}\in \DSort_n({\geq} 2;123)$.  We analyze the two cases in the
  definition of $\varphi$ separately.

  If $\pi \in \DSort_{n-1}(1;123)$ (that is $\pi_1=n-1$), then, by
  Lemma~\ref{starting_descent}, $s_{123}(\pi )=(n-2)(n-3)\cdots 1(n-1)$.  Now
  we analyze what happens with input $\tilde{\pi}$ after the first pass
  through the (restricted) stack.  Remember that the first element of
  $\tilde{\pi}$ is $n-1$ and that $n$ immediately follows $n-2$;
  moreover, suppose that $n-2$ is the $i$-th element of $\tilde{\pi}$.
  Therefore, the first $i$ elements of $\pi$ and $\tilde{\pi}$ are
  equal, and so they are processed exactly in the same way by the stack.
  In particular, since $n-2$ is the first element of $s_{123}(\pi )$,
  when $n-2$ enters the stack, all the previous elements of
  $\tilde{\pi}$ are still inside the stack.  Immediately after $n-2$
  enters the stack, $n$ enters the stack as well, since it cannot
  produce a forbidden pattern inside the stack.  Now we claim that $n$
  and $n-2$ exit the stack together.  In fact, if $n$ is not the last
  element of $\tilde{\pi}$, consider the next element $\pi_{i+1}$.  Such
  an element cannot enter the stack, otherwise $n-1$ (which is at the
  bottom of the stack), $n-2$ and $\pi_{i+1}$ would constitute a
  forbidden pattern.  Thus $n-2$ must exit the stack before $\pi_{i-1}$
  enters it, and this forces $n$ to exit as well.  As a consequence of
  this fact, we have that
  $s_{123}(\tilde{\pi})=n(n-2)(n-3)\cdots 1(n-1)$, and such a
  permutation does not contain the pattern 231.  Hence $\tilde{\pi}$ is
  sortable.

  If $\pi \in \DSort_{n-1}({\geq} 2;123)$ (that is $\pi_1 =k\neq n-1$), then
  $s_{123}(\pi )=(n-1)(n-2)\cdots (k+1)(k-1)\cdots 21k$, and an
  analogous argument proves that
  $s_{123}(\tilde{\pi})=n(n-1)(n-2)\cdots (k+1)(k-1)\cdots 21k$, and so
  that $\tilde{\pi}$ is sortable.

  To complete the proof we now have to show that $\varphi$ is a
  bijection.  The fact that $\varphi$ is injective is trivial.  To show
  that $\varphi$ is surjective, consider the map
  $\psi: \DSort_n({\geq} 2;123)\rightarrow \DSort_{n-1}(123)$ which removes $n$
  from
  $\alpha =\alpha_1 \alpha_2 \cdots \alpha_n \in \DSort_n({\geq} 2;123)$.
  Set $\psi (\alpha )=\hat{\alpha}$.  Let $i\in \{ 3,4,\ldots n\}$ such
  that $\alpha_i =n$.  From Corollary~\ref{before_maximum} we have that
  either $\alpha_{i-1}=n-1$ (if $\alpha_1 \neq n-1$) or
  $\alpha_{i-1}=n-2$ (if $\alpha_1 =n-1$).  Moreover, Lemma~\ref{starting_descent} implies that
  $s_{123}(\pi)=n(n-1)\cdots (k+1)(k-1)\cdots 21k$, with $k=\alpha_1 \geq 2$.
  Therefore, when $n$ enters the stack, all the previous elements are
  still inside the stack.  In particular, at the top of the stack there
  are $n$ and $\alpha_{i-1}$.  Now notice that, if $n$ is forced to exit
  the stack, this is due to the fact that there exist $j,h,l$, with
  $j<h\leq i$ and $l>i$, such that $\alpha_j ,\alpha_h$ and $\alpha_l$
  form an occurrence of 321.  However, it cannot be $h=i$, since $n$
  cannot play the role of the 2 in a 321.  Similarly, it cannot be
  $h=i-1$: in fact, if $\alpha_{i-1}=n-1$, then $n$ and $n-1$ are
  consecutive in the stack and so they play the same role in any
  pattern; if instead $\alpha_{i-1}=n-2$, then $\alpha_1 =n-1$ is at the
  bottom of the stack, and so $n$ and $n-2$ play the same role in any
  forbidden pattern.  As a consequence, $h<i-1$, and so $n$ and
  $\alpha_{i-1}$ are forced to leave the stack together.  This means
  that basically $n$ does not modify the behavior of the machine, and so
  $s_{123}(\hat{\alpha})=(n-1)(n-2)\cdots (k+1)(k-1)\cdots 21k$, that is
  $\hat{\alpha}$ is sortable, as desired.
\end{proof}

\begin{corollary}
  For all $n\geq 3$, $|\DSort_n({\geq} 2;123)|=|\DSort_{n-1}(123)|$.
\end{corollary}

What we have proved so far, and in particular Corollary~\ref{inflation},
Corollary~\ref{starting_maximum} and Theorem~\ref{bijection}, completely
determine the structure of sortable permutations.  Indeed, any
$\pi\in\Sort_n(123)$ which is not the identity permutation can be uniquely
constructed as follows:

\begin{itemize}
\item choose $\alpha =\alpha_1 \alpha_2 \cdots \alpha_k \in \Av_k (213)$,
  with $\alpha_1 =k\geq 2$;
\item add $h$ new maxima, $k+1,\ldots k+h$, one at a time, using the
  bijection $\varphi$ of Theorem~\ref{bijection};
\item add $n-k-h$ consecutive ascents at the beginning, by inflating the
  first element of the permutation, according to Corollary~\ref{inflation}.
\end{itemize}

As an example to illustrate the above construction, consider the
permutation $\pi =567148923$.  By deflating the starting consecutive
ascents we get the permutation 5146723; we then observe that the last
permutation is obtained by adding two new maxima to the permutation
$51423\in \Av(213)$ according to the bijection of Theorem~\ref{bijection}.
Since 51423 starts with its maximum, we can conclude
that $\pi$ is sortable.

The above construction allows us to enumerate $\Sort_n(123)$.

\begin{theorem}
  For all $n\geq 1$,
  $$
  |\Sort_n(123)|=1+\sum_{h=1}^{n-1}(n-h)C_h ,
  $$
  where $C_n =\frac{1}{n+1}{2n\choose n}$ is the $n$-th Catalan number.
\end{theorem}

\begin{proof}
  A permutation $\pi \in \Sort_n(123)$ is either the identity or it is
  obtained by choosing a permutation $\alpha$ of $\Av_k (213)$ starting
  with its maximum $k$ (with $k\geq 2$) and then (possibly) adding the
  remaining $n-k$ elements according to the above construction,
  i.e. adding new maxima and/or some consecutive ascents at the
  beginning.  Concerning $\alpha$, there are $C_{k-1}$ possible choices,
  thanks to the observation following Corollary~\ref{starting_maximum}.
  Concerning the remaining elements, one has to choose, for instance,
  the number of new maxima to add, which runs from 0 to $n-k$, so that
  the total number of choices is $n-k+1$.  Summing on all possible
  values of $k$, we get:
  $$
  |\Sort_n(123)|=1+\sum_{k=2}^{n}C_{k-1}\cdot (n-k+1)=1+\sum_{h=1}^{n-1}(n-h)C_h ,
  $$
  as desired.
\end{proof}

From the above closed form it is not difficult to find the generating function of $|\Sort_n(123)|$.
However, we prefer to use a bijective argument.
The number sequence $(|\Sort_n(123)|)_{n\in \mathbb{N}}$ is sequence A294790 in \cite{Sl}.
A combinatorial interpretation of this sequence can be found in \cite{CF}:
it counts the number of Schr\"oder paths avoiding the pattern $\U\H_2\D$.

We say that a Schr\"oder path $P$ avoids the pattern $\U\H_2\D$ when
$P$ does not contain three steps that, read from left to right, are $\U$, $\H_2$ and $\D$, respectively.
In \cite{CF} it is observed that Schr\"oder paths avoiding $\U\H_2\D$ are precisely
those of the form $\H_2^\alpha Q\H_2^\beta$, where $Q$ is a Dyck path.

We now describe a bijection $f$ between
sortable permutations of length $n$ and $\U\H_2\D$-avoiding Schr\"oder paths of semilength $n-1$.
Given $\pi\in\Sort_n(123)$, we decompose it as $\pi =Lw$, where $L$ is the
(possibly empty) initial sequence of consecutive ascents of $\pi$
deprived of the last element and $w$ is the remaining suffix of $\pi$.
Suppose that $L$ has length $r$.  Now repeatedly remove the maximum from
$w$ until the remaining word $v$ starts with its maximum.  Denote with
$s$ the number of elements removed in such a way.  The permutation
obtained from $v$ after rescaling is then a 213-avoiding permutation
$\rho$ starting with its maximum of length $k+1=n-r-s$.  Removing the
maximum from $\rho$ results in another 213-avoiding permutation $\sigma$
of length $k$.
We can now describe the Schr\"oder path $f(\pi )$ associated with $\pi$:
it starts with $r$ double horizontal steps and ends with $s$ double
horizontal steps; in the middle, there is the Dyck path of semilength
$k$ associated with the 213-avoiding permutation $\sigma$ through the
bijection described in Section~\ref{prelim}.  For instance, referring to
the above notations, given the permutation $\pi=567489132$, we have
$L=56$, $w=7489132$, and so $r=s=2$.  Moreover, $\sigma =4132$, and so
the Dyck path associated with $\sigma$ is $\U\D\U\U\D\U\D\D$.  Finally, we thus
have that $f(\pi )=\H_2 \H_2 \U\D\U\U\D\U\D\D \H_2 \H_2$ (see Figure~\ref{Schroder}).

\begin{figure}[h!]
\begin{center}
\begin{tikzpicture}[scale=0.7]
\draw [ultra thin] (-2,0) -- (14,0);
\draw [thick] (-2,0) -- (2,0);
\draw [thick] (2,0) -- (3,1);
\draw [thick] (3,1) -- (4,0);
\draw [thick] (4,0) -- (6,2);
\draw [thick] (6,2) -- (7,1);
\draw [thick] (7,1) -- (8,2);
\draw [thick] (8,2) -- (10,0);
\draw [thick] (10,0) -- (14,0);
\node[] at (-2,0) {$\bullet$};
\node[] at (0,0) {$\bullet$};
\node[] at (2,0) {$\bullet$};
\node[] at (3,1) {$\bullet$};
\node[] at (4,0) {$\bullet$};
\node[] at (5,1) {$\bullet$};
\node[] at (6,2) {$\bullet$};
\node[] at (7,1) {$\bullet$};
\node[] at (8,2) {$\bullet$};
\node[] at (9,1) {$\bullet$};
\node[] at (10,0) {$\bullet$};
\node[] at (12,0) {$\bullet$};
\node[] at (14,0) {$\bullet$};

\node [above left] at (-2,0) {$\mathbf{5}$};
\node [above left] at (0,0) {$\mathbf{6}$};
\node [above left] at (2,0) {$\mathbf{7}$};

\node [above left] at (2.5,0.5) {$\mathbf{4}$};
\node [above left] at (4.5,0.5) {$\mathbf{1}$};
\node [above left] at (5.5,1.5) {$\mathbf{3}$};
\node [above left] at (7.5,1.5) {$\mathbf{2}$};
\node [below left] at (3.5,0.7) {$4$};
\node [below left] at (6.5,1.5) {$3$};
\node [below left] at (8.5,1.5) {$2$};
\node [below left] at (9.5,0.7) {$1$};

\node [above] at (11,0) {$\mathbf{8}$};
\node [above] at (13,0) {$\mathbf{9}$};

\end{tikzpicture}
\caption{A Schr\"oder path avoiding $UH_2 D$;
the associated permutation is obtained by reading the bold labels according to the bijection described above.}
\label{Schroder}
\end{center}
\end{figure}
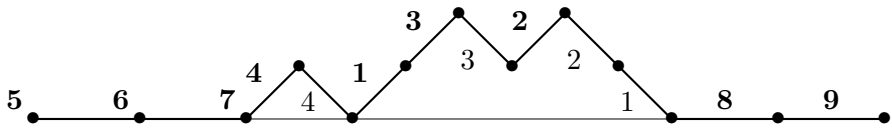

Now, as announced, we express the generating function of
$(|\Sort_n(123)|)_{n\in \mathbb{N}}$ by exploiting the above bijection.  In
fact, the generic Schr\"oder path avoiding $\U\H_2\U$ either consists of
double horizontal steps only (so the generating function is
$(1-x)^{-1}$), or it can be obtained by concatenating an initial
sequence of double horizontal steps (having generating function
$(1-x)^{-1}$) with a nonempty Dyck path (whose generating function is
$(C(x)-1)\cdot x$, where $C(x)$ is the generating function of Catalan
numbers and the additional factor $x$ takes into account the removal of
the starting maximum from the permutation $\rho$ above), finally adding
a sequence of double horizontal steps (again with generating
function $(1-x)^{-1}$).  Summing up, we get:
$$
\sum_{n\geq 0}|\Sort_n(123)|x^n
=\frac{1}{1-x}+\frac{1}{1-x}\big(x(C(x)-1)\big)\frac{1}{1-x}
=\frac{(1-x)^2}{1-2x+xC(x)}.
$$

\section{Further work}\label{conclusion}

Our paper just scratches the surface of the subject, and opens the way
to the investigation of sorting permutations using restricted stacks in
series, where the restrictions are expressed in terms of pattern
avoidance.

Along the lines of this work, concerning patterns of length 3, there are
four cases yet to study (corresponding to the four nonmonotone
patterns). This is something we are presently working on.  We have some
data concerning the enumeration of $\sigma$-sortable permutations of
length $n$, for small values of $n$.  In the following table we report
our data, including the (conjectured) references to \cite{Sl}, where
applicable.

{\scriptsize
  \begin{center}
    \bgroup
    \def\arraystretch{1.5}
    \begin{tabular}{c|rrrrrrrrrrrr|c}
      $\sigma\setminus n$ & 0 & 1 & 2 & 3 & 4 & 5 & 6 & 7 & 8 & 9 & 10 & 11 & OEIS \\
      \hline
      132 & 1&1&2&5&15&51&188&731&2950&12235&51822&223191 & A294790  \\
      213 & 1&1&2&5&16&62&273&1307&6626&35010&190862&1066317 & unknown \\
      231 & 1&1&2&6&23&102&496&2569&13934&78295&452439&2674769 & unknown \\
      312 & 1&1&2&5&15&52&201&843&3764&17659&86245&435492 & A202062 \\
    \end{tabular}
    \egroup
  \end{center}
}

Concerning longer patterns, it would be interesting to classify
$\sigma$-machines in terms of the number of permutations they sort.
This gives rise to a notion of \emph{Wilf-equivalence} on
$\sigma$-machines.  This seems to be particularly interesting when the
set of $\sigma$-sortable permutations constitute a class.  For instance,
for patterns of length 4, there are precisely two Wilf-equivalence
classes, corresponding to two types of $\sigma$-machines, which depend
on the number of sortable permutations: the resulting sequences are
Catalan numbers and (essentially) odd-indexed Fibonacci numbers
(sequence A001519 in \cite{Sl}).

\bibliography{nonclasses}

\begin{thebibliography}{10}

\bibitem{BM}
M.~Bousquet-M\'elou.
\newblock Discrete excursions.
\newblock {\em Sem. Lothar. Combin.}, 57:23 pp., 2008.

\bibitem{CF}
L.~Cioni and L.~Ferrari.
\newblock Enumerative results on the {S}chr\"oder pattern poset.
\newblock In {\em Cellular Automata and Discrete Complex Systems, AUTOMATA
  2017, Lecture Notes in Comput. Sci., vol 10248}, pages 56--67, 2017.

\bibitem{CLF}
J.~A. Clapperton, P.~J. Larcombe, and E.~J. Fennessey.
\newblock On iterated generating functions for integer sequences, and catalan
  polynomials.
\newblock {\em Util. Math.}, 77:3--33, 2008.

\bibitem{GX}
I.~Gessel and G.~Xin.
\newblock A combinatorial interpretation of the numbers 6(2n)!/n!(n + 2)!
\newblock {\em J. Integer Seq.}, 8:Article 05.2.3, 13 pp., 2005.

\bibitem{Kn}
D.~E. Knuth.
\newblock {\em The Art of Computer Programming, Volume 1}.
\newblock Boston: Addison-Wesley, 1968.

\bibitem{Kr}
C.~Krattenthaler.
\newblock Permutations with restricted patterns and dyck paths.
\newblock {\em Adv. in Appl. Math.}, 27:510--530, 2001.

\bibitem{M}
M.~M. Murphy.
\newblock {\em Restricted permutations, antichains, atomic classes and stack
  sorting}.
\newblock PhD thesis, University of St Andrews, 2002.

\bibitem{Sl}
N.~J.~A. Sloane.
\newblock {\em The On-line Encyclopedia of Integer Sequences}.

\bibitem{Sm}
R.~Smith.
\newblock Two stacks in series: a decreasing stack followed by an increasing
  stack.
\newblock {\em Ann. Comb.}, 18:359--363, 2014.

\bibitem{W}
J.~West.
\newblock {\em Permutations with forbidden subsequences and Stack sortable
  permutations}.
\newblock PhD thesis, Massachusetts Institute of Technology, 1990.

\end{thebibliography}
\bibliographystyle{plain}

\end{document}